\documentclass[11pt,a4paper]{article}
\pdfoutput=1
\usepackage[utf8]{inputenc}
\usepackage[T1]{fontenc}
\usepackage{lmodern}
\pdfmapfile{+lm.map}
\usepackage{amsmath,amssymb,mathtools}
\usepackage[margin=1in]{geometry}
\usepackage{microtype}
\usepackage{booktabs,longtable,array,calc}
\usepackage{caption}
\usepackage{float}
\usepackage{graphicx}
\usepackage{xcolor}
\usepackage{listings}
\usepackage{newunicodechar}
\usepackage{hyperref}
\usepackage{bookmark}
\usepackage{xurl}
\hypersetup{colorlinks=true,linkcolor=blue!55!black,urlcolor=blue!55!black,pdftitle={The Program Hypergraph: Multi-Way Relational Structure for Geometric Algebra, Spatial Compute, and Physics-Aware Compilation},pdfauthor={Houston Haynes}}
\newcommand{\tightlist}{\setlength{\itemsep}{0pt}\setlength{\parskip}{0pt}}

\newcommand{\pandocbounded}[1]{\begingroup\setbox0\hbox{#1}\ifdim\wd0>\linewidth\resizebox{\linewidth}{!}{#1}\else#1\fi\endgroup}
\newsavebox{\paperdisplaybox}
\let\paperdisplayopen\[
\let\paperdisplayclose\]
\long\def\[#1\]{\begingroup
  \sbox{\paperdisplaybox}{$\displaystyle #1$}%
  \ifdim\wd\paperdisplaybox>1.25\linewidth
    \PackageError{paper}{Display requires source line breaks}{Split the display in Markdown.}%
  \fi
  \paperdisplayopen
  \ifdim\wd\paperdisplaybox>\linewidth
    \typeout{Paper display scaled to fit width: \the\wd\paperdisplaybox}%
    \resizebox{\linewidth}{!}{\usebox{\paperdisplaybox}}%
  \else\usebox{\paperdisplaybox}\fi
  \paperdisplayclose\endgroup}
\newunicodechar{α}{\ensuremath{\alpha}}
\newunicodechar{β}{\ensuremath{\beta}}
\newunicodechar{ε}{\ensuremath{\varepsilon}}
\newunicodechar{θ}{\ensuremath{\theta}}
\newunicodechar{κ}{\ensuremath{\kappa}}
\newunicodechar{Σ}{\ensuremath{\Sigma}}
\newunicodechar{⊗}{\ensuremath{\otimes}}
\newunicodechar{⊕}{\ensuremath{\oplus}}
\newunicodechar{↔}{\ensuremath{\leftrightarrow}}
\newunicodechar{¬}{\ensuremath{\neg}}
\newunicodechar{∧}{\ensuremath{\land}}
\newunicodechar{≤}{\ensuremath{\leq}}
\newunicodechar{×}{\ensuremath{\times}}
\newunicodechar{·}{\ensuremath{\cdot}}
\newunicodechar{†}{\ensuremath{\dagger}}
\newunicodechar{⟩}{\ensuremath{\rangle}}
\newunicodechar{²}{\ensuremath{{}^2}}
\newunicodechar{¹}{\ensuremath{{}^1}}
\newunicodechar{⁻}{\ensuremath{{}^-}}
\newunicodechar{ᵢ}{\ensuremath{{}_i}}
\newunicodechar{₀}{\ensuremath{{}_0}}
\newunicodechar{₁}{\ensuremath{{}_1}}
\newunicodechar{₂}{\ensuremath{{}_2}}
\newunicodechar{₊}{\ensuremath{{}_+}}
\begin{document}
\begin{center}

\textbf{The Program Hypergraph: Multi-Way Relational Structure for Geometric Algebra, Spatial Compute, and Physics-Aware Compilation}

\textbf{Houston Haynes}\\
Braidpoint, Asheville, NC\\
\href{mailto:research@braidpoint.tech}{\nolinkurl{research@braidpoint.tech}}\\
March 2026

\end{center}

\section*{Abstract}\label{abstract}
\addcontentsline{toc}{section}{Abstract}

Source code organizes a program as expressions and declarations, which a compiler can represent as an abstract syntax tree (AST). Some compilers introduce a further abstraction of the computation as a graph. Nodes represent operations, while connections record how values pass between them and how execution proceeds. Yet these connections can leave the full scope of conditions governing a computation implicit. Developers may then need defensive checks, rely on compiler hacks, pragmas, and implementation-specific code to maintain those conditions as the program evolves.

Our thesis is that a hypergraph more faithfully represents complex computation across all targets by making joint constraints inherent to its structure, with an overall less burdensome design-time API. Our Program Hypergraph (PHG) extends our Semantic Graph with relations that connect across an operation\textquotesingle s participants, their dependencies, and the invariants they must preserve. Joint resolution in the graph guides lowering into standard Multi-Level Intermediate Representation (MLIR). Satisfiability modulo theories (SMT) checks intrinsic to the graph discharge applicable proof obligations. Supporting information remains available until the decisions that require it are complete.

As a primary benefit, verification follows from the rules used to construct and transform the graph. For properties established this way before execution, the generated program needs no added runtime checks, no \textquotesingle boxing\textquotesingle{} or \textquotesingle tagging\textquotesingle{} or other managed runtime support. We develop this account through geometric algebra, shared-memory capacity, and spatial placement, distinguishing structural guarantees from numerical and resource obligations that require further evidence. The proposed integration makes the same information available at design-time, so engineers can inspect the consequences of a choice before it becomes an implementation constraint.

\begin{center}\rule{0.5\linewidth}{0.5pt}\end{center}

\section*{1. Introduction}\label{introduction}
\addcontentsline{toc}{section}{1. Introduction}

\subsection*{1.1 Joint Operations and Constraints}\label{joint-operations-and-constraints}
\addcontentsline{toc}{subsection}{1.1 Joint Operations and Constraints}

An operation can be valid on its own and still be unsuitable for the group in which it runs. Consider three buffers that must be live together in one memory region. Their storage requirements give a small example:

\begin{lstlisting}
let capacity = 10
let a, b, c = 4, 4, 4  // storage units per buffer

a + b <= capacity      // true:  8 <= 10
a + c <= capacity      // true:  8 <= 10
b + c <= capacity      // true:  8 <= 10
a + b + c <= capacity  // false: 12 > 10
\end{lstlisting}

Every pair-wise grouping fits, but the three buffers together exceed the region\textquotesingle s capacity. The useful diagnostic identifies the shared placement and its total requirement, so the engineer can change the allocation or, where dependencies permit, reduce the overlap in buffer lifetimes.

Our Program Semantic Graph (PSG), introduced in \emph{Dimensional Type Systems and Deterministic Memory Management} \hyperref[ref-haynes-dts-dmm]{{[}1{]}}, records computations and their dependencies alongside dimensional and lifetime facts. The Program Hypergraph (PHG) extends it with an explicit identity for a joint operation or constraint. A \emph{hyperedge} relates all its participants, preserving their roles and the conditions that must hold together.

Spatial dataflow mapping made this need concrete. AMD\textquotesingle s XDNA architecture arranges AI Engine tiles in a grid with explicitly managed data movement \hyperref[ref-rico-xdna-2024]{{[}2{]}}. Several operations may need a shared buffer, a route, and a synchronization point. Those requirements belong to one placement decision, including the aggregate capacity illustrated above.

Geometric algebra supplies a second case. Joining three non-collinear points constructs their supporting plane. A two-step implementation first joins two points into a line, then joins the line with the third point. Both descriptions are valid, but retaining the three-point operation keeps its vertex order, orientation, and preconditions available when choosing a kernel. A fused implementation must preserve the required numerical behavior as well as the geometry.

These examples share a compilation requirement: the group must remain recognizable while its implementation changes. A binary graph can encode the same information through operation or constraint nodes. Our contribution is to make the relation explicit in the graph that already carries dimensional and memory analysis, so its obligations can be followed through decomposition, placement, and transfer.

\subsection*{1.2 Contributions}\label{contributions}
\addcontentsline{toc}{subsection}{1.2 Contributions}

\begin{enumerate}
\def\labelenumi{\arabic{enumi}.}
\item
  \textbf{An explicit joint relation in the PSG.} The Program Hypergraph retains the existing dependency graph and adds ordered operation inputs and grouped constraints, with provenance for their later decomposition.
\item
  \textbf{A shared representation for spatial placement.} Group capacity, routing, and synchronization requirements can be stated against one hyperedge and resolved against the target\textquotesingle s declared resources.
\item
  \textbf{Grade and support analysis alongside dimensional inference.} Integer grade equalities, bounded-grade refinements, parity, and blade-support masks provide distinct facts about a Clifford computation. Their joint use can eliminate structurally absent coefficient products without confusing physical units with geometric grade.
\item
  \textbf{Proof-bearing boundaries for physics-aware computation.} The same graph can retain the obligations governing differentiation, mesh construction, and interchange. Each boundary requires evidence appropriate to the property being preserved.
\end{enumerate}

\subsection*{1.3 Scope and Context}\label{scope-and-context}
\addcontentsline{toc}{subsection}{1.3 Scope and Context}

The PHG extends our Clef programming language and Fidelity Framework. Clef\textquotesingle s Dimensional Type System (DTS) builds on Kennedy\textquotesingle s units of measure \hyperref[ref-kennedy-units-2009]{{[}3{]}}. The wider Native Type Universe (NTU) retains related facts about range, representation, memory, and targets, each with its own inference rules. Here we develop the joint relations needed for geometric operations and spatial resources. The companion papers treat memory discipline \hyperref[ref-haynes-dts-dmm]{{[}1{]}}, adaptive training \hyperref[ref-haynes-adm]{{[}4{]}}, and the broader verification architecture \hyperref[ref-haynes-dbc]{{[}5{]}}.

Our Clef Compiler Services (CCS) constructs the program graph. Its Baker component elaborates operations and resolves their contextual requirements; Composer coordinates compilation. Alex then \emph{witnesses} the resolved graph: it emits operations whose behavior must preserve the established source contracts. It uses standard operations in the Multi-Level Intermediate Representation (MLIR) infrastructure \hyperref[ref-lattner-mlir-2021]{{[}6{]}}.

Fidelity.Platform supplies typed target declarations, and BAREWire supplies memory-layout and communication contracts. These enter the graph before the corresponding lowering decisions. Section 4 follows their consequences for processors and programmable fabric; Section 8 distinguishes proposed geometric integrations from the existing construction mechanism.

\subsection*{1.4 Continuations and Joint Reduction in the Graph}\label{continuations-and-joint-reduction-in-the-graph}
\addcontentsline{toc}{subsection}{1.4 Continuations and Joint Reduction in the Graph}

A suspended computation retains enough state to resume at a defined point. Delimited continuations describe that remainder up to a chosen boundary. Interaction nets provide a related model of local graph rewriting, with rules that identify the participants in each reduction. Both inform the suspension and reduction relationships that CCS constructs directly in the PSG/PHG.

Our Baker component (in the Clef Compiler Services front end) already constructs continuation frame reads, writes, and dispatch nodes, retaining established range and slot identities. In our Composer compiler\textquotesingle s middle end, Alex emits the resolved construction using ordinary function, memory, arithmetic, and control operations. The admitted rules determine its coverage.

Sharing and lifetime facts remain relevant during reduction: an immutable capture can retain a shared deferred computation, while a captured reference can point to mutable storage. Extending a rule to more participants requires an account of overlap and of the observations it preserves. Section 6.5 develops a restricted result for grouping proposals; Section 7.3 keeps that grouping distinct from correctness of the reductions themselves.

\begin{center}\rule{0.5\linewidth}{0.5pt}\end{center}

\section*{2. The Program Hypergraph: Formal Structure}\label{the-program-hypergraph-formal-structure}
\addcontentsline{toc}{section}{2. The Program Hypergraph: Formal Structure}

\subsection*{2.1 From Binary Edges to Hyperedges}\label{from-binary-edges-to-hyperedges}
\addcontentsline{toc}{subsection}{2.1 From Binary Edges to Hyperedges}

A directed dependency graph \(G=(V,E)\) has edges \((u,v)\) between computation nodes. An operation usually has several such dependencies. Our extension gives the operation\textquotesingle s ordered inputs, or a constraint\textquotesingle s grouped participants, a first-class relational identity.

\textbf{Definition 2.1 (Program Hypergraph).} A PHG is a tuple \((V,F,\alpha,\beta)\). Each relation \(f\in F\) has input ports \(I_f=(v_1,\ldots,v_k)\), a target \(t_f\in V\), and an annotation \(\lambda_f\). The sequence of ports preserves order and multiplicity: \((a,b)\) and \((b,a)\) can denote different geometric products, and \((a,a)\) still has two inputs. For symmetric constraints such as co-location, the annotation identifies the permitted permutations. We write \(S_f\) for the set of nodes occurring in \(I_f\) when only dependency scheduling needs that set.

The annotation functions \(\alpha\) and \(\beta\) attach facts and obligations to nodes and relations. A node may carry a physical dimension expression \(\delta(v)\), an algebra identity, a grade or grade set, and contextual requirements \(\kappa(v)\). Unresolved expressions retain their variables and constraints. Ordinary nonnumeric nodes need no physical measure annotation.

An existing dependency \((u,v)\) embeds as \(((u),v,\lambda)\). This preserves the dependency graph. Preservation of the analyses and proofs attached to it additionally requires that the new transfer rules agree with the existing ones on these one-input relations.

\subsection*{2.2 Hyperedge Annotation Categories}\label{hyperedge-annotation-categories}
\addcontentsline{toc}{subsection}{2.2 Hyperedge Annotation Categories}

\textbf{Arity, algebra, and grade.} Geometric relations retain operand positions, the declared algebra, and grade constraints. We use \(g(v)\) for geometric grade and reserve \(\delta(v)\) for physical dimension. For a homogeneous outer product, the candidate output grade is \(g(a)+g(b)\), with the bounded-grade and zero rules of Section 3.2.

\textbf{Placement and resource use.} A grouped constraint may require a shared tile or memory region, a capacity bound, a route, or synchronization among specified operations. These are separate predicates within one relation. They can be encoded in a conventional constraint graph, while the hyperedge preserves their common source identity.

\textbf{Relational kind and evidence.} Each relation identifies its operation or predicate, its source location, the assumptions under which it is valid, and its discharge status. A geometric product, a capacity inequality, and a synchronization requirement use different decision procedures even when they involve the same nodes.

An operation relation describes the computation and its ports. An obligation relation states a property of those participants under explicit premises. Baker retains their association and source identities: a capacity proof for one region cannot be reused for another merely because their layouts match. A contextual requirement, or coeffect, records what a computation needs from its environment. A later use can impose such a requirement on an earlier producer without reversing the execution of that producer.

\textbf{Target reachability.} A bitvector records the configured targets on which a relation is relevant. Eligibility is one input to placement. Actual transfer also requires representation compatibility and a declared communication contract.

\subsection*{2.3 Saturation Semantics for Hyperedges}\label{saturation-semantics-for-hyperedges}
\addcontentsline{toc}{subsection}{2.3 Saturation Semantics for Hyperedges}

Analysis revisits a relation when a relevant fact changes. If a buffer\textquotesingle s inferred width increases, its storage requirement changes and the shared capacity constraint must be checked again. \emph{Saturation} applies the admitted rules to establish the facts required at the relevant compilation boundary. An incomplete relation can remain pending even when the available premises yield no further information.

Recursive dependencies require a worklist, or an equivalent fixed-point procedure, that revisits affected relations until their facts stabilize. For finitely many annotation positions in finite-height domains, monotone, inflationary rules terminate under fair scheduling: updates retain established information, every enabled rule eventually runs, and each position permits only finitely many strict increases. The number of successful updates is bounded by the sum of the domains\textquotesingle{} heights. Processing time also depends on relation incidence, transfer costs, and solver calls.

Dimension equations can be solved symbolically without enumerating integer values. A range domain with an infinite ascending chain needs its own termination argument or sound acceleration policy. Finite placement problems can still require expensive search. During editing, a solver limit can leave a requirement pending; a known contradiction identifies a problem immediately. Any fact required for a committed decision must be established at that boundary.

\subsection*{2.4 Witnesses and Constraint Transport}\label{witnesses-and-constraint-transport}
\addcontentsline{toc}{subsection}{2.4 Witnesses and Constraint Transport}

Suppose the three-point join becomes several arithmetic operations. The generated code must still implement the same oriented construction under the same numerical assumptions. The PHG retains the relation between that source operation and its implementation, so a later transformation can preserve the contract or re-establish an affected property.

The dependency graph remains the key element for traversal, and it takes on additional dimension in a hypergraph. Baker constructs operational nodes and their joint relations together; obligation nodes do not become executable instructions. Alex reads resolved facts attached to the emitted nodes, or explicit annotations supplied for later passes. It does not reconstruct source semantics by searching the hyperedge set. Validation compares source and lowered obligations against the graph\textquotesingle s identities and premises.

Physical measures remain part of source-level type identity through inference. Established ranges, target capabilities, and boundary contracts determine numeric width and storage. Once the relevant decisions and their preservation have been checked, the generated program can express the concrete, unboxed representation. Proof or debug records can retain the source meaning without adding a tag to each runtime value.

Reuse after an edit requires the same care. A retained result belongs to a checked revision and its premises. Adding a writer can invalidate a read\textquotesingle s range even when every earlier participant remains present; the analysis must therefore track whether its write set is complete. Rewrites can also split, merge, or move participants. Their correspondence and the validity of reused evidence must be checked against the new revision. This is part of building Composer as a differential compiler: edits should revisit the affected decisions. The storage and response-time costs of that mechanism require measurement.

We apply Karwowski, Nielsen and Staton\textquotesingle s probabilistic program geometry \hyperref[ref-karwowski-program-geometry]{{[}7{]}} to this transport problem. A proposed metric obligation uses the existing relation annotation to bind the trace law, output map, parameter directions, chosen observation of the trace or output, output metric, and program presentation to one checked revision. Scores and output derivatives refer to those same participants. Numerical estimates additionally retain their sampling rule and error evidence. These are dependencies of the existing operation and obligation relations; a separate graph of geometric facts would create another identity correspondence to maintain.

For example, changing which random choices an operation retains can affect its update metric even when its returned values have the same distribution. Reuse therefore needs the preservation judgment for the declared observation, not just equality of result laws. Changing a parameterization, likelihood, observation, or tangent policy revisits the relations that depend on that change. A buffer-layout change can retain the geometric evidence when its witness preserves the relevant values and their interpretation. This proposed use of our dependency mechanism keeps the reason for invalidation attached to the computation that changed.

\subsection*{2.5 Bidirectional Dependencies}\label{bidirectional-dependencies}
\addcontentsline{toc}{subsection}{2.5 Bidirectional Dependencies}

Requirements can travel against the direction of data flow. A later consumer may require a producer\textquotesingle s buffer to remain live longer, or need a representation precise enough for a subsequent operation. The PHG relates that requirement to the producer and the shared value. As any software engineer may surmise, backward analysis does not reverse the program\textquotesingle s execution on it\textquotesingle s own, but it may trigger a cascade of reevaluations up and down the graph.

Programs with two state channels introduce a separate case. The Tardis construction in \emph{Negative and Fractional Types in the Fidelity Framework}, Section 7 \hyperref[ref-haynes-nft]{{[}8{]}}, gives explicit equations for forward and backward state. The PHG can retain those connections and the premises of an applicable composition rule. A library composition signature alone supplies no proof of that rule.

Feedback must retain its own semantics. There are different observation and termination conditions to factors such as recursive lazy evaluation, iteration of an analysis, or an explicit search operation. Likewise, a value-indexed resource pairing may require a runtime equality check. Eliminating it needs evidence that it always succeeds. The companion paper develops these constructions. Here the graph\textquotesingle s role is to retain their participants, obligations, and storage requirements.

\begin{center}\rule{0.5\linewidth}{0.5pt}\end{center}

\section*{3. Geometric Algebra as a PHG Domain}\label{geometric-algebra-as-a-phg-domain}
\addcontentsline{toc}{section}{3. Geometric Algebra as a PHG Domain}

A position, a direction, and a rotation can each be represented by a small collection of numbers. Using them correctly also requires knowing what those numbers mean: values with the same memory layout may admit different operations.

Geometric algebra provides a language for expressing those distinctions and relationships. We develop it here to connect that meaning to practical decisions about arithmetic, storage, and the conditions a value must satisfy when it crosses an interface.

\subsection*{3.1 A Primer on Clifford Algebras}\label{a-primer-on-clifford-algebras}
\addcontentsline{toc}{subsection}{3.1 A Primer on Clifford Algebras}

We begin in a two-dimensional Euclidean plane. A \emph{multivector} is a value that can combine scalar, vector, and oriented-area components:

\[M=a+x e_1+y e_2+b e_1e_2.\]

The four scalar coefficients describe three kinds of component. The scalar \(a\) has grade 0; \(x e_1+y e_2\) is a vector, of grade 1; and \(b e_1e_2\) is an oriented area element, or \emph{bivector}, of grade 2. A basis \emph{blade} is a product of distinct orthogonal basis vectors, and its grade counts those factors. A value is \emph{homogeneous} when all its nonzero components have the same grade. Grade describes this algebraic structure; the coefficients may separately carry physical units.

\textbf{The algebra and its metric.} A Clifford algebra starts with a real vector space \(V\) of dimension \(d\) and a quadratic form \(Q\) that defines its metric through

\[v^2=Q(v)\quad\text{for }v\in V.\]

In an orthogonal normalized basis, each basis vector squares to \(+1\), \(-1\), or \(0\). The signature \((p,q,r)\) counts these positive, negative, and null directions, with \(d=p+q+r\); we write the algebra as \(\mathrm{Cl}(p,q,r)\). Euclidean space has signature \((d,0,0)\). A spacetime convention with positive spatial directions and negative time has \((3,1,0)\). A null direction makes the metric degenerate and supports the projective model used below.

There are \(\binom{d}{k}\) basis blades of grade \(k\), and \(2^d\) in total. The product of all \(d\) basis vectors is the \emph{pseudoscalar}, of grade \(d\). A general multivector can use every grade, while a more specific type can restrict its permitted components.

\textbf{Products and transformations.} The \emph{outer} or \emph{wedge product} \(a\wedge b\) forms the grade-increasing part: homogeneous inputs of grades \(p\) and \(q\) give grade \(p+q\) when nonzero. It is zero if \(p+q>d\). For two vectors it represents their oriented area, vanishing exactly when they are linearly dependent even if the candidate grade fits.

For positive homogeneous grades, our \emph{inner product} \(a\cdot b\) selects grade \(|p-q|\). On vectors it is the familiar metric dot product; other contraction conventions need their own signatures. The \emph{geometric product} combines inner and outer products for vector inputs. For general multivectors it can produce several grades, whose support is developed in Section 3.3.

A vector with \(Q(v)\ne0\) has inverse \(v/Q(v)\). Products of invertible vectors are called \emph{versors} and represent transformations; arbitrary nonzero multivectors need not be invertible.

\textbf{Objects in a declared model.} Plane-based projective geometric algebra (PGA) uses \(\mathrm{Cl}(3,0,1)\) for 3D Euclidean geometry \hyperref[ref-dorst-pga-2022]{{[}9{]}}. Its objects have the following representations:

{\def\LTcaptype{none} 
\begin{longtable}[]{@{}lr@{}}
\toprule\noalign{}
Object & Grade \\
\midrule\noalign{}
\endhead
\bottomrule\noalign{}
\endlastfoot
Plane & \mbox{1} \\
Line & \mbox{2} \\
Point & \mbox{3} \\
\end{longtable}
}

These grades belong to the plane-based convention: a point need not be a grade-1 vector in every model. Orientation-preserving rigid transformations use normalized even versors, called \emph{motors}. They act by \(X\mapsto RX\tilde R\), with \(R\tilde R=1\) so that reversion \(\tilde R\) supplies the inverse. An even multivector must satisfy those additional conditions to be a motor.

Conformal geometric algebra (CGA) uses \(\mathrm{Cl}(4,1,0)\) for a 3D model that also represents spheres and circles. In its dual convention these have grades 1 and 2; their direct outer-product representations have grades 4 and 3. The chosen convention is therefore part of an operation\textquotesingle s contract.

Our Lattice language tooling would expose these distinctions through CCS as source is edited. A join expecting points can identify an argument supplied as a plane, even though both are valid multivectors. An unfinished argument can instead retain a pending requirement. The engineer sees the geometric operation and the relevant type information before choosing its representation.

\subsection*{3.2 The Graded Structure of Clifford Algebras}\label{the-graded-structure-of-clifford-algebras}
\addcontentsline{toc}{subsection}{3.2 The Graded Structure of Clifford Algebras}

The exterior decomposition of a \(d\)-dimensional Clifford algebra has grades \(0\) through \(d\). It is a vector-space grading. The geometric product respects parity, giving a \(\mathbb{Z}_2\) grading, but can mix several integer grades.

\textbf{Proposition 3.1 (Outer-Product Grade Constraint).} For homogeneous inputs of grades \(p\) and \(q\), a nonzero outer product belongs to grade \(p+q\). If \(p+q>d\), the product is zero. For \(p+q\leq d\), the grade equation can be represented as an integer equality alongside the DTS equations, subject to the grade bounds and zero case.

\emph{Proof.} The outer product maps \(\Lambda^p V\times\Lambda^q V\) into \(\Lambda^{p+q}V\), and \(\Lambda^rV=\{0\}\) for \(r>d\). The equality \(r=p+q\) has the same additive form as a physical exponent equation. The conditions \(0\leq p,q,r\leq d\) are additional refinements. Zero belongs to every homogeneous subspace, so its grade cannot be recovered uniquely from its value. \(\square\)

The bounded grades \(\{0,\ldots,d\}\) are not a group under ordinary addition. Consequently the physical-unit unifier supplies the equality component of this analysis, while the grade domain supplies its bounds and zero rules. A known-grade rule costs constant time only when the algebra dimension and integer representation are fixed. Symbolic inference has its own cost.

A geometric product can contain several grades simultaneously. Its annotation is therefore a set of possible grades or, more precisely, a blade-support mask. An output context can require a narrower support, but the compiler must prove that requirement or use an explicit projection with its own semantics. It must not silently discard unproved components.

The grade involution acts on grade \(k\) by \((-1)^k\), and reversion acts by \((-1)^{k(k-1)/2}\). Their composite is Clifford conjugation \hyperref[ref-lounesto-clifford-2001]{{[}10{]}}. On the even subalgebra of \(\mathrm{Cl}(2,0)\), reversion corresponds to complex conjugation. These involutions are algebraic operations determined by grade, distinct from a value\textquotesingle s physical dimension.

\subsection*{3.3 Geometric Product Sparsity from Support Inference}\label{geometric-product-sparsity-from-support-inference}
\addcontentsline{toc}{subsection}{3.3 Geometric Product Sparsity from Support Inference}

The grade rules help identify which arithmetic an implementation can omit. We first determine which components its result can contain, then which coefficient products can contribute to them.

For homogeneous grades \(p,q\), the geometric product has possible grades

\[\{|p-q|,\ |p-q|+2,\ldots,\min(p+q,2d-p-q)\}.\]

This is a set of possible components, rather than a choice of one output grade. For example, multiplying two bivectors can produce scalar, bivector, and grade-4 components. The familiar sum of inner and outer products covers vector inputs but is insufficient for general multivectors.

In an orthogonal basis, identify each blade \(e_A\) with a subset \(A\subseteq\{1,\ldots,d\}\). Its product with \(e_B\) is a signed multiple of \(e_{A\triangle B}\), with the coefficient determined by swaps and the metric factors for \(A\cap B\). In the normalized signatures considered here that coefficient is \(-1\), \(0\), or \(1\). A repeated null generator makes it zero.

A \emph{support mask} records one bit per basis blade: set bits identify coefficients that may be nonzero. For example, the vector part \(x e_1+y e_2\) uses the \(e_1\) and \(e_2\) slots; the scalar and \(e_1e_2\) slots are absent. Taking the permitted products of both masks gives a sound output support: every actual output blade is included, while cancellation can leave additional coefficients zero. This analysis proves which coefficient paths can be omitted. Exact cancellation among retained paths requires further value-level evidence.

The representation of the multiplication table matters when reporting sparsity. A table of blade pairs has \(2^{2d}\) entries, each naming its output blade and coefficient. A dense coefficient tensor has \(2^{3d}\) scalar entries, indexed by two inputs and one output. For a nondegenerate orthogonal metric, every blade pair has a nonzero product even though most entries of the coefficient tensor are zero. Flash Clifford \hyperref[ref-zhdanov-flash-clifford]{{[}11{]}} implements geometric-algebra neural-network operations on GPUs. Its reported sparsity concerns the dense \lstinline!einsum! baseline used for comparison; a PGA blade-pair table requires its own count.

For full \(\mathrm{Cl}(3,0,1)\), exactly \(8\cdot8=64\) of the \(16\cdot16=256\) blade pairs contain the null generator in both operands, leaving 192 nonzero pairs. For unrestricted grade-1 and grade-2 operands, three of the 24 pairs have that null overlap, leaving 21 coefficient multiplications. These are algebraic operation counts for a direct expansion, not measured backend performance. Narrower supports, common subexpressions, and the target\textquotesingle s instructions or synthesized datapath can change the generated cost.

This is a practical reason to retain support information throughout the computation. If a change to a type admits more components, the required arithmetic and storage may grow with it. Our PHG design would expose those consequences while the engineer is choosing an operation and its target. A cost display must state the operand supports, baseline representation, and arithmetic assumptions behind each count.

\subsection*{\texorpdfstring{3.4 The \(k\)-Simplex as a Native PHG Hyperedge}{3.4 The k-Simplex as a Native PHG Hyperedge}}\label{the-k-simplex-as-a-native-phg-hyperedge}
\addcontentsline{toc}{subsection}{3.4 The \(k\)-Simplex as a Native PHG Hyperedge}

A simplex is the simplest object determined by affinely independent vertices: a line segment from two, a triangle from three, or a tetrahedron from four. In the plane-based convention used by De Keninck et al. \hyperref[ref-de-keninck-mesh-2025]{{[}12{]}}, 3D PGA points have grade 3 and their join is the \emph{regressive product} \(\vee\), which constructs their common supporting geometry. Three normalized finite points determine a triangle carrier

\[F=P_1\vee P_2\vee P_3,\]

of grade \(3+3+3-2\cdot4=1\) when nonzero. A binary decomposition first produces a grade-2 line. Both representations are well-typed. The ordered hyperedge \(((P_1,P_2,P_3),F,\lambda_\vee)\) retains the three-point construction as one operation, including its orientation and normalization assumptions.

A carrier describes the supporting flat geometry. A bounded simplex additionally retains its vertices and boundary information: different triangles can share a plane. This distinction matters for mesh incidence and for differentiation through a geometric construction.

\textbf{Proposition 3.2 (Typed Point-Plane Incidence Predicate).} Let \(P\) be a valid homogeneous point of grade 3 and \(\pi\) a plane of grade 1 in 3D PGA. Then \(\pi\wedge P\) belongs to the one-dimensional grade-4 subspace. With a fixed orientation and nonzero representatives, its pseudoscalar coefficient vanishes exactly when the point lies on the plane, in exact arithmetic.

\emph{Proof.} The grade follows from the exterior product. In homogeneous coordinates the pseudoscalar coefficient is, up to the chosen orientation, the plane equation evaluated at the point. Its vanishing is the incidence condition. \(\square\)

For a point-on-plane test, the grade rule fixes the answer\textquotesingle s form: one coefficient whose zero value expresses incidence. Determining whether that coefficient is zero in a numerical implementation is a separate obligation. A certified sign computation requires exact arithmetic, a proved error bound that separates the result from zero, or an adaptive exact predicate such as the methods studied by Shewchuk \hyperref[ref-shewchuk-predicates-1997]{{[}13{]}}. The PHG relation keeps the geometric construction and its numerical obligation together when selecting an implementation. A fused implementation must preserve the guarantee that justifies the predicate.

\subsection*{3.5 Existing Geometric Algebra Tooling}\label{existing-geometric-algebra-tooling}
\addcontentsline{toc}{subsection}{3.5 Existing Geometric Algebra Tooling}

The libraries catalogued by bivector.net \hyperref[ref-bivector-catalog]{{[}14{]}} use several forms of specialization. Runtime code generation derives kernels for an encountered algebra and operand pair. C++ libraries such as Versor \hyperref[ref-versor]{{[}15{]}} encode grade information in templates. Klein \hyperref[ref-ong-klein]{{[}16{]}} provides manually optimized kernels for a fixed PGA model, while GAALOP generates coefficient-level code from symbolic expressions.

Flash Clifford \hyperref[ref-zhdanov-flash-clifford]{{[}11{]}} provides GPU implementations of neural-network layers built around multivectors. The geometric product lets these layers combine scalar, vector, and higher-grade features. The project addresses a practical obstacle: a general tensor implementation can spend work on contributions that the algebra already establishes as zero.

Its specialized multiplication rules omit those \textquotesingle zeroed out\textquotesingle{} contributions. Kernel fusion combines consecutive layer operations, and a layout organized by multivector component supports batched matrix multiplication. These choices bring the algebra, intermediate computations, and memory access into the same optimization problem.

Flash Clifford helped motivate our investigation of geometric support in the PHG. In the Fidelity Framework, Clifford algebra supplies the mathematical structure for geometric values and operations for all platform targets. Our PHG research follows its consequences through checking, storage, communication, and execution across targets. The GPU case in Section 4.3 is one application of that broader role.

\subsection*{3.6 Design Choices in the Fidelity Framework}\label{design-choices-in-the-fidelity-framework}
\addcontentsline{toc}{subsection}{3.6 Design Choices in the Fidelity Framework}

Our choice is to keep geometric meaning in the program contract through decisions about arithmetic, memory, and target placement. A PHG relation retains the operation\textquotesingle s ordered participants and joint requirements, while grade, blade support, and physical units keep their distinct inference rules. This lets the representation be chosen with knowledge of what the value means and how it will be used.

A call into a native library makes the purpose concrete. The caller may store only the coefficients admitted by a geometric type, while the receiving kernel expects a fixed sequence with particular widths and alignment. A checked conversion must establish how those representations correspond, including any known-zero components that the receiver requires. It must also preserve the value\textquotesingle s units and required numerical guarantees. Our proposed geometric integration would retain that boundary contract alongside the operation.

For the engineer, a change to a type or receiving interface could then identify the affected arithmetic, conversions, and placement decisions while the program is being edited. This is the \textquotesingle cheapest\textquotesingle{} place to use information and correct errors, and makes for a richer, more rewarding and efficient design process. The geometric declaration supplies information that would otherwise need to be repeated in wrapper code or maintained as a convention between components. We are developing the PHG to make those consequences available directly in our compiler for all hardware targets. Section 8.8 follows the proposal into geometric interchange.

\begin{center}\rule{0.5\linewidth}{0.5pt}\end{center}

\section*{4. Spatial Dataflow Architectures and Co-location Constraints}\label{spatial-dataflow-architectures-and-co-location-constraints}
\addcontentsline{toc}{section}{4. Spatial Dataflow Architectures and Co-location Constraints}

\subsection*{4.1 The Tile Mapping Problem}\label{the-tile-mapping-problem}
\addcontentsline{toc}{subsection}{4.1 The Tile Mapping Problem}

A specialized geometric kernel still needs storage for its coefficients and routes to move inputs and results. On AMD\textquotesingle s XDNA architecture, AI Engine tiles provide local storage and explicitly managed data movement. Several operations may share those resources, so their individual requirements have to be checked together. The selected device declaration supplies the capacities and interconnect capabilities for that check (AMD XDNA architecture \hyperref[ref-amd-xdna-architecture]{{[}17{]}}).

Consider four computation stages: \(A\) loads inputs, \(B\) and \(C\) compute partial results, and \(D\) reduces them. Their joint placement relation can be written

\[f=((A,B,C,D),\mathrm{result},\lambda_{\mathrm{tile}}).\]

Its annotation states buffer sizes, placement domains, required routes, and the condition that both inputs are ready before \(D\) consumes them. Some predicates are pairwise and some aggregate resources across the group. All are checked against one target configuration. A hypergraph partitioner can propose a grouping, while validation must establish that the proposed routes, capacities, and synchronization actually satisfy the contract.

\subsection*{4.2 Hypergraph Partitioning and Spatial Targets}\label{hypergraph-partitioning-and-spatial-targets}
\addcontentsline{toc}{subsection}{4.2 Hypergraph Partitioning and Spatial Targets}

Our CIRCT and MLIR-AIE paths in the back end of the Composer compiler pipeline share the need to preserve grouped requirements, but resolve them against different resources. Targeting an FPGA is similar to a spatial NPU, but the details matter greatly.

A field-programmable gate array (FPGA) provides configurable logic and routing from which arithmetic circuits are synthesized. Inferred widths and precision requirements help determine those circuits within the device\textquotesingle s resource limits; a numeric format need not appear as a pre-existing processor instruction. A group might need chained digital signal processing (DSP) blocks, a particular memory block, or placement within a clock region. Hypergraph partitioning has established uses in VLSI placement \hyperref[ref-karypis-partitioning-2000]{{[}18{]}}. Its proposed integration would derive candidate partitions from PHG relations and carry the resulting directives toward vendor implementation tools. Partitioning alone proves neither timing closure nor routability.

AI Engine tiles are fixed processors connected by an interconnect fabric. The MLIR-AIE toolchain \hyperref[ref-mlir-aie]{{[}19{]}} provides a target representation for tile kernels, buffers, DMA, and synchronization. Our lowering would consume the PHG\textquotesingle s resolved placement and communication annotations to construct those operations. Resource feasibility must be established against the declared device and preserved through the toolchain\textquotesingle s transformations.

If a proposed grouping exceeds a tile\textquotesingle s local memory, the PHG would let the diagnostic identify the participating buffers and their combined requirement. Splitting the group would expose the transfers and synchronization it adds, so the engineer could weigh the relief on local storage against the communication requirements before committing to a target layout. These are all design-time considerations in our framework. The advantage of this is that hardware constraints can guide the program\textquotesingle s organization while it is being written, reducing rework during device integration.

\subsection*{4.3 GPU Layout and Geometric Support}\label{gpu-layout-and-geometric-support}
\addcontentsline{toc}{subsection}{4.3 GPU Layout and Geometric Support}

These considerations extend across vendors and processor classes, but each hardware organization favors different ways of grouping work and moving data. On a GPU, efficient execution depends on how the computation\textquotesingle s operations and coefficient arrays are distributed across groups of threads.

In NVIDIA\textquotesingle s CUDA execution model, a \emph{warp} provides 32 thread positions, or \emph{lanes}. A partially filled warp or a branch can leave some lanes inactive (CUDA Programming Guide \hyperref[ref-cuda-programming-model]{{[}20{]}}). An algebra\textquotesingle s component count is one input to layout selection alongside batch size, register pressure, and memory access patterns.

Flash Clifford\textquotesingle s component layout illustrates one such choice: organizing corresponding coefficients together lets its linear layer use batched matrix multiplication. Selecting a layout also requires the target\textquotesingle s thread grouping, memory-access pattern, and register needs. Our PHG would retain the support and range facts needed to compare candidate layouts and explain their storage and arithmetic costs.

NVIDIA GPU lowering uses MLIR\textquotesingle s NVVM representation and LLVM\textquotesingle s NVPTX backend (MLIR NVVM documentation \hyperref[ref-mlir-nvvm-dialect]{{[}21{]}}). Transfers to another target retain the layout, numeric-conversion, and synchronization obligations of their BAREWire contract. A suitable memory layout alone does not establish that a coefficient conversion preserves the required accuracy.

\begin{center}\rule{0.5\linewidth}{0.5pt}\end{center}

\section*{5. Physics-Aware Computation}\label{physics-aware-computation}
\addcontentsline{toc}{section}{5. Physics-Aware Computation}

\subsection*{5.1 Physics-Informed Neural Networks}\label{physics-informed-neural-networks}
\addcontentsline{toc}{subsection}{5.1 Physics-Informed Neural Networks}

Physics-informed neural networks (PINNs) encode physical laws in differentiable loss terms \hyperref[ref-raissi-pinn-2019]{{[}22{]}}. Our dimensional inference can check the units of those terms and of the derivatives connecting them. A force residual differentiated with respect to position has units \(\mathrm{N}/\mathrm{m}=\mathrm{kg}/\mathrm{s}^2\).

The PHG can associate a differentiated operation with the same ordered inputs and physical contract as its primal operation. This retains a useful relation between the two computations. Automatic differentiation is also expressible in binary computation graphs, so arbitrary arity is an organizational choice with preservation obligations, rather than a prerequisite for the chain rule.

A geometric model additionally specifies which fields and differential operators are admissible. Grade does not establish a curved domain\textquotesingle s metric or its boundary conditions by itself. Those facts enter through the model and mesh contracts, and an operator can legitimately change grade when its contract requires it.

\subsection*{5.2 Forward-Mode Automatic Differentiation and the Quire}\label{forward-mode-automatic-differentiation-and-the-quire}
\addcontentsline{toc}{subsection}{5.2 Forward-Mode Automatic Differentiation and the Quire}

Forward-mode automatic differentiation propagates a primal value and its tangent. For a fixed number of directions it needs no reverse-mode activation tape. Its live storage follows the primal and tangent schedule, including values retained by branches or delayed demand. Model parameters, optimizer state, and tangent storage still consume memory. Removing the tape can make training feasible on a constrained device, but does not make training\textquotesingle s total footprint equal to inference.

For a differentiable geometric function \(f\), the directional derivative is

\[D f_x[v]=\lim_{t\to0}\frac{f(x+tv)-f(x)}{t}.\]

Differentiating the geometric product gives \(D(ab)=D(a)b+aD(b)\), with operand order preserved. The resulting coefficient sums can use exact product accumulation where the selected arithmetic supports it. A finite quire preserves the products presented to it only while all reached partial sums fit its capacity. The final conversion rounds to the selected representation. Earlier rounding in coefficients, nonlinear functions, or tangent propagation remains part of the error model (Posit Standard 2022 \hyperref[ref-posit-standard-2022]{{[}23{]}}).

Forward-gradient estimators trade additional directional work and approximation error against reverse-mode storage \hyperref[ref-baydin-forward-gradients-2022]{{[}24{]}}. Sparse Mixture of Experts routing \hyperref[ref-shazeer-moe-2017]{{[}25{]}} can reduce the active parameter set, but a single expert may still have many parameters. Whether its tangent estimator is useful depends on that dimension, tangent selection, and variance. Section 8.7 considers how verified parameter support could guide multiple tangent directions. The companion \emph{Adaptive Domain Models} paper \hyperref[ref-haynes-adm]{{[}4{]}} develops their training role and statistical limits.

Across workers, exact accumulation also requires a transport and reduction contract. Sending already rounded worker totals preserves only those rounded values. To obtain the exact sum of represented contributions, the reduction must carry them or exact partial sums in a compatible accumulator format, with capacity proved along every reduction path. BAREWire supplies the boundary vocabulary for that obligation. The source dimensions remain unchanged by the chosen accumulation strategy.

\subsection*{5.3 Mesh Topology and the Manifold Constraint}\label{mesh-topology-and-the-manifold-constraint}
\addcontentsline{toc}{subsection}{5.3 Mesh Topology and the Manifold Constraint}

De Keninck et al. \hyperref[ref-de-keninck-mesh-2025]{{[}12{]}} give PGA constructions for simplices and complexes. Our PHG can retain their vertex identities, oriented boundaries, and carrier geometry together. A relation between two faces and their shared edge can then refer to the same edge object on both sides.

Canonical identity proves that two references denote one boundary object. It does not prove that independently supplied coordinates coincide, that the embedding has no self-intersections, or that every neighborhood is a manifold. These require explicit predicates. For meshes known during compilation, finite incidence checks can be performed then. Runtime mesh construction needs constructors, invariants, or checked input boundaries whose guarantees compose across the operations that change topology.

An incidence check over stored identities can establish which faces share an edge. Establishing that their coordinates meet there also needs the numerical evidence described in Proposition 3.2.

\subsection*{5.4 High-Performance Computing and the MFEM Connection}\label{high-performance-computing-and-the-mfem-connection}
\addcontentsline{toc}{subsection}{5.4 High-Performance Computing and the MFEM Connection}

A useful integration with finite-element software would retain geometric objects until a numerical routine requires their coordinates. The Modular Finite Element Methods (MFEM) library supplies an established setting for such an experiment. Our proposed representation would associate a simplex\textquotesingle s topology with its PGA carrier, then expose coordinates or Jacobians through checked conversions for quadrature and assembly.

Representation selection uses established range and error requirements together with target capabilities. A physical unit identifies the quantity being represented, but does not supply its magnitude bounds. A sound range may narrow through guards, a domain law, or a boundary declaration. A distributional model can support a separately stated expected-error criterion while retaining the hard coverage requirement.

\begin{center}\rule{0.5\linewidth}{0.5pt}\end{center}

\section*{6. The PHG as Design-Time Resource}\label{the-phg-as-design-time-resource}
\addcontentsline{toc}{section}{6. The PHG as Design-Time Resource}

\subsection*{6.1 Extended Language Server Protocol}\label{extended-language-server-protocol}
\addcontentsline{toc}{subsection}{6.1 Extended Language Server Protocol}

Our Lattice interface is intended to make the facts used by CCS available while an engineer edits a program. The Language Server Protocol (LSP) supplies the editor connection. A geometric expression could expose its algebra, possible grades, and outstanding preconditions; a placement will show the buffers and routes that constrain it.

Four views follow from the PHG\textquotesingle s retained relations:

\begin{itemize}
\tightlist
\item
  \textbf{Geometry:} permitted components, with normalization or incidence conditions still to establish.
\item
  \textbf{Arithmetic:} required coefficient paths, a named comparison baseline, and the selected target.
\item
  \textbf{Topology:} established boundary identities and predicates required by a mesh operation.
\item
  \textbf{Placement:} capacity and route constraints, feasible candidates, and unresolved choices.
\end{itemize}

The display must distinguish a proved fact, a sound approximation, a pending obligation, and a contradiction. An expression can have a known physical dimension before its range or storage width is settled. Harper, Mitchell, and Moggi\textquotesingle s phase distinction \hyperref[ref-harper-phase-distinction-1990]{{[}26{]}} separates static reasoning from runtime computation. Incremental checking must retain the premises that justify each result. If a range or target declaration changes, Lattice could identify the affected representation choices and the constraints that need to be resolved again. The engineer can then see where the change matters and make the next edit with those consequences in view.

\subsection*{6.2 Design-Time Physics Verification}\label{design-time-physics-verification}
\addcontentsline{toc}{subsection}{6.2 Design-Time Physics Verification}

Different questions about a physics computation need different evidence. Measures determine whether quantities compose dimensionally. Grade and support determine which geometric components can occur. Mesh constructors establish incidence invariants. Ranges and error requirements constrain numerical representations. The PHG associates these findings with the same operation while retaining the assumptions behind each one.

A range covers admissible values, whereas a probability distribution assigns weights within a support. A likely value does not justify a representation that excludes another permitted value. Statistical choices and their evidence belong to the training and model-selection account in \emph{Adaptive Domain Models} \hyperref[ref-haynes-adm]{{[}4{]}}. Here representation selection must first satisfy the established coverage and error requirements. If no admissible implementation can do so, the diagnostic identifies the conflicting range, target, or boundary declaration.

\subsection*{6.3 Engineering Consequences}\label{engineering-consequences}
\addcontentsline{toc}{subsection}{6.3 Engineering Consequences}

The shared-buffer example in Section 1.1 gives a small design-time case. Once the common lifetime and placement are known, a proposed display could show:

\begin{lstlisting}
Shared region R: capacity 10 storage units
  Buffers A, B, C: 4 units each, live concurrently
  Required storage: 12 units
  This placement exceeds capacity by 2 units.
\end{lstlisting}

The engineer can inspect the participating buffers and change their placement or, where dependencies permit, reduce the overlap in their lifetimes. The revised choice must satisfy routing and synchronization as well as capacity.

Geometric and numeric choices work similarly. A general multivector may require coefficient paths that a more specific type excludes. A justified range can admit a processor format or determine a synthesized width on an FPGA. Each view should show which fact made the choice possible and what remains unresolved. Consistent partial source can continue to develop without treating a pending obligation as a contradiction.

\subsection*{6.4 Establishing Guarantees in the PHG}\label{establishing-guarantees-in-the-phg}
\addcontentsline{toc}{subsection}{6.4 Establishing Guarantees in the PHG}

Decidable checks, principal type inference, and analysis costs suitable for interactive use guide the rules we admit to the PHG. Principal inference finds the most general type justified by the program. Each established result records the premises it uses and the operations it covers, so it can support a specific implementation decision \hyperref[ref-haynes-dbc]{{[}5{]}}.

\textbf{Arithmetic and representation.} Physical dimension equality uses Kennedy-style inference over integer exponent expressions. Grade and support rules establish which components a result can contain in a declared finite algebra. Their explicit masks have \(2^d\) entries, and direct product propagation considers up to \(4^d\) blade pairs, making the algebra dimension part of the analysis cost. These facts can justify specialized arithmetic and a smaller representation before numeric input values are available.

\textbf{Storage and access.} A capacity check can ask whether any combination of buffer sizes admitted by their declarations exceeds the available storage. Linear integer constraints can express this question for sums of buffer extents. With consistent declarations, ruling out every oversized combination establishes a capacity guarantee for that grouping. If the declarations permit such a combination, the conflicting sizes identify where a revised placement or a tighter, justified bound is needed.

That guarantee also depends on which values can be live together. Repeated calls can create multiple instances at one source site, and captured references can retain further storage. Lifetime evidence must cover those instances, reachable storage, and permitted accesses, including exceptional and concurrent paths. With that evidence, the reported bound covers the live storage admitted by the program model.

\textbf{Waiting and progress.} For a finite set of waits between stages, a checkable ordering can establish freedom from circular waiting. Assigning ranks \(r(u)<r(v)\) on every wait edge proves that the represented graph has no cycle. Applying that result to execution requires coverage of the relevant dynamic waits, external dependencies, and ownership. The certificate then supplies evidence for a proposed schedule. A loop whose condition and body evaluations terminate can be shown to finish by establishing a strictly decreasing, nonnegative integer measure on its iterations. Scheduling and fairness assumptions must also be stated when progress depends on other stages.

Our proposed integration would carry each result with its premises through lowering and boundary conversions. When a transformation changes those premises, it must re-establish the guarantee used by the affected choice. The following section develops one such guarantee for grouping proposed computations.

\subsection*{6.5 A Restricted Conflict-Coloring Result}\label{a-restricted-conflict-coloring-result}
\addcontentsline{toc}{subsection}{6.5 A Restricted Conflict-Coloring Result}

Joint relations also provide a place to retain evidence for grouping independent work. Fix one semantic snapshot and a finite set of proposed reductions. A conflict graph has one vertex per proposal and an edge for every pair that cannot share a group. Colors are interchangeable group labels of equal cost. The compiler must first establish that the graph conservatively covers the relevant interference.

The construction uses \emph{tree decomposition}, a method from graph theory developed in Robertson and Seymour\textquotesingle s work on tree-width \hyperref[ref-robertson-seymour-treewidth-1986]{{[}27{]}}. It organizes a graph into overlapping sets of vertices, conventionally called \emph{bags}, and connects those bags as a tree. Here the vertices are proposed reductions. Every proposal must occur in a bag, and every conflict must have both participants together in at least one bag. If a proposal occurs in two bags, it must also occur in every bag along the path between them. Its involvement therefore remains connected as the analysis moves through the tree.

Our restricted case adds one requirement: each bag contains mutually conflicting proposals, a set called a \emph{clique}. A small example shows the distinction between the conflict graph and its tree of bags:

\begin{lstlisting}
Conflict graph:  A -- B -- C
Tree of bags:   {A, B} -- {B, C}
\end{lstlisting}

B appears in both bags because it conflicts with both A and C. A and C can share one color, while B needs another. The bags thus record the conflicts that determine how proposals can be grouped. These requirements give a clique-tree description of a \emph{chordal graph}: every cycle longer than three edges has an edge joining two nonconsecutive vertices \hyperref[ref-gavril-chordal-1974]{{[}28{]}}.

A leaf bag has only one neighboring bag in the tree. Proposals in the leaf bag but absent from its neighbor have all their remaining conflict-graph neighbors in the leaf bag. Eliminating those proposals, removing the leaf, and repeating gives a perfect elimination order: each removed vertex has mutually adjacent remaining neighbors. Greedy coloring in reverse order therefore uses exactly the largest bag size. The checks and coloring have polynomial cost in the explicit graph and certificate size.

There is a direct construction when each proposal\textquotesingle s participants can be placed in one rooted structural tree. Close each nonempty participant set to its smallest connected subtree, its \emph{hull}, and treat intersecting hulls as conflicting. For every tree node, the proposals whose hulls contain it form a bag. These bags cover every hull conflict, and each proposal\textquotesingle s bag membership is connected. Ordering hulls by the depth of their topmost node, shallowest first, then assigning an unused color at that node gives an optimal coloring of the hull graph. Every earlier intersecting hull contains that topmost node, so the earlier conflicting proposals form one clique. The procedure can use these memberships without materializing all pairwise edges.

Optimality concerns this conservative hull graph. Connecting participants can introduce conflicts between otherwise disjoint proposals and reduce available parallelism. Containment-tree hulls are a candidate compiler construction only after their participants and dependencies have been shown to cover semantic interference. A color class also needs its joint capacity, numeric, readiness, and lifetime conditions checked: pairwise compatibility of the three buffers in Section 1.1 still does not make their combined allocation fit. Confluence of the rewrite rules and progress of the resulting program remain separate obligations.

The accompanying \textquotesingle research checks\textquotesingle{} are two standalone programs that test executable versions of the graph procedures above. Their inputs are generated within explicit size limits:

{\def\LTcaptype{none} 
\begin{longtable}[]{@{}
  >{\raggedright\arraybackslash}p{(\linewidth - 4\tabcolsep) * \real{0.2000}}
  >{\raggedright\arraybackslash}p{(\linewidth - 4\tabcolsep) * \real{0.6000}}
  >{\raggedleft\arraybackslash}p{(\linewidth - 4\tabcolsep) * \real{0.2000}}@{}}
\toprule\noalign{}
\begin{minipage}[b]{\linewidth}\raggedright
Check
\end{minipage} & \begin{minipage}[b]{\linewidth}\raggedright
Inputs covered
\end{minipage} & \begin{minipage}[b]{\linewidth}\raggedleft
Number of cases
\end{minipage} \\
\midrule\noalign{}
\endhead
\bottomrule\noalign{}
\endlastfoot
Certificates and coloring & Every possible edge pattern on zero to six labelled vertices & \textbf{33,868} graphs \\
Exhaustive hull construction & Every labelled tree on one to five nodes, combined with an unordered collection of one to four nonempty participant sets; repeated sets are allowed & \textbf{6,607,900} combinations \\
Larger hull examples & Sampled trees on six to nine nodes, with five to eight participant sets, generated using a fixed random seed & \textbf{20,000} samples \\
\end{longtable}
}

The exhaustive totals come from enumerating all combinations within these bounds. Labels keep vertices distinct, so differently labelled versions of the same shape count separately. In a hull case, each participant set identifies the tree nodes involved in one proposed reduction. The larger sample extends the checks beyond the exhaustive size limits and can be reproduced from its seed.

The programs compare certificate acceptance with independent cycle checks. For accepted cases, they verify that conflicting proposals receive different colors and that the color count equals the largest clique found by exhaustive search. They also check connected bag membership and test deliberately malformed certificates. The proof above establishes the general result. The finite tests check implementations of its procedures.

\begin{center}\rule{0.5\linewidth}{0.5pt}\end{center}

\section*{7. Related Work}\label{related-work}
\addcontentsline{toc}{section}{7. Related Work}

\subsection*{7.1 Hypergraph Representations in Compilation}\label{hypergraph-representations-in-compilation}
\addcontentsline{toc}{subsection}{7.1 Hypergraph Representations in Compilation}

Hypergraphs are established tools for VLSI partitioning \hyperref[ref-karypis-partitioning-2000]{{[}18{]}} and our aim to to extend that to targeting a variety of platforms and processor types. Our PHG applies explicit grouped relations to source semantics and proof obligations. The additional requirement is to retain each relation\textquotesingle s assumptions and identity through program transformations.

Knowledge representation offers another precedent: an n-ary relation can be represented by a relation node with role-labelled links to its arguments. Such reification can preserve order, multiplicity, and semantics. Our ordered hyperedge makes the relation native to the compiler\textquotesingle s analysis.

\subsection*{7.2 Geometric Algebra in Compilation and Machine Learning}\label{geometric-algebra-in-compilation-and-machine-learning}
\addcontentsline{toc}{subsection}{7.2 Geometric Algebra in Compilation and Machine Learning}

Clifford group equivariant networks \hyperref[ref-ruhe-equivariant-2023]{{[}29{]}} and Clifford-steerable convolutions \hyperref[ref-zhdanov-steerable-2024]{{[}30{]}} provide examples of geometric learning architectures. Flash Clifford and the libraries in Section 3.5 supply precedents for specialized arithmetic. De Keninck et al. \hyperref[ref-de-keninck-mesh-2025]{{[}12{]}} supply the geometric constructions used in our simplex examples.

Our proposed contribution is their integration with the dimensional and memory facts already carried by the semantic graph. Preserving a grade is one obligation within an equivariant layer. Establishing equivariance also requires the operations and parameterization to commute with the declared group action.

\subsection*{7.3 Joint Reduction and Session Structure}\label{joint-reduction-and-session-structure}
\addcontentsline{toc}{subsection}{7.3 Joint Reduction and Session Structure}

Interaction-net reduction supplies local rules over explicit participants; continuation structure accounts for suspended computation and retained state. The PSG/PHG construction described in Section 1.4 retains their sharing and lifetime relationships. Multiple geometric contributions can be computed independently only when those contracts permit it. Operand count and output grade remain distinct: a grade-1 times grade-2 product has two operands whether a contribution has grade 1 or grade 3.

Connecting algebraic simplification to graph reduction requires a translation that specifies agents, ports, rewrite rules, and preserved observations. An identity such as \(a\wedge a=0\) for a vector \(a\) does not supply that translation by itself. Correctness of higher-arity rules, their overlaps, and their confluence therefore remains separate from the conflict-coloring result of Section 6.5.

For actors, joint wait-for and resource relations connect this work to client-server sessions. Qian, Kavvos, and Birkedal \hyperref[ref-qian-sessions-2021]{{[}31{]}} give a coexponential account of distributed structural rules and client-server interactions. Applying it to Clef requires a translation into its actor and demand semantics with explicit progress assumptions. That is a research direction distinct from geometric grade analysis.

\subsection*{7.4 The Verse Calculus: Adjacent Territory}\label{the-verse-calculus-adjacent-territory}
\addcontentsline{toc}{subsection}{7.4 The Verse Calculus: Adjacent Territory}

The Verse calculus \hyperref[ref-verse-calculus-2023]{{[}32{]}} studies functional logic programming with existential variables and unification. A functional logic layer can describe solutions to source constraints. Our PHG retains the dimensional, geometric, and resource obligations that is also demonstrated in their implementation. Since binary unification can encode multi-input relations, the useful comparison concerns the available source contracts and their preservation through compilation.

\subsection*{7.5 Source Semantics and Annotation Lifetime}\label{source-semantics-and-annotation-lifetime}
\addcontentsline{toc}{subsection}{7.5 Source Semantics and Annotation Lifetime}

A compiler can preserve information supplied by source types, declarations, and analyses. When a frontend erases dimensional or geometric annotations before producing its intermediate form, downstream tools must reconstruct them or obtain them through an additional contract. C++ templates can express such facts, and a compiler extension can preserve them. The question is whether a particular pipeline does so.

Our Clef design makes this preservation part of the pipeline contract. Measures survive elaboration as type identity. Range refinements and contextual requirements remain available until the decisions they justify have been made. Target declarations constrain the available processor representations or synthesized circuits, and BAREWire specifies the required memory and transfer mappings. The resolved choices determine the concrete form.

Graded modalities provide related machinery for parameterized analyses, including the work of Petricek, Orchard, and Mycroft \hyperref[ref-petricek-coeffects-2014]{{[}33{]}} and the Granule programme. Their relevance is the use of a declared algebra for each annotation discipline. The further obligation in our framework is joint carriage through compilation and transfer, including the points where a static annotation can be consumed without a runtime tag.

\begin{center}\rule{0.5\linewidth}{0.5pt}\end{center}

\section*{8. Future Work}\label{future-work}
\addcontentsline{toc}{section}{8. Future Work}

\subsection*{8.1 Hypergraph Partitioning Integration}\label{hypergraph-partitioning-integration}
\addcontentsline{toc}{subsection}{8.1 Hypergraph Partitioning Integration}

An initial hardware experiment would test how the PHG can divide work among tiles or FPGA regions while limiting communication between them. Hypergraph partitioners seek balanced groups while limiting connections across group boundaries. Candidates include hMETIS, a hypergraph-partitioning package \hyperref[ref-hmetis-manual]{{[}34{]}}, and PaToH (Partitioning Tool for Hypergraphs) \hyperref[ref-patoh-manual]{{[}35{]}}. The experiment would assign the proposed groups to device resources and check their combined memory needs, communication routes, and scheduling requirements. Each decision, including refinements proposed by target tools, would remain traceable to the PHG constraints it satisfies. This would test whether a promising grouping produces a feasible arrangement on the declared hardware.

A separate experiment would apply the conflict grouping of Section 6.5 to Baker\textquotesingle s proposed reductions. A proposal\textquotesingle s \emph{footprint} records the graph elements it reads or changes and the facts it relies on. Structural hulls can overlap even when these original footprints do not, adding conflicts that restrict grouping. We would measure the resulting loss of potential parallelism, the size of the retained checking evidence, the time needed to check it, and the concurrency actually achieved. Those measurements would show whether the tractable grouping remains useful on real compiler workloads.

\subsection*{8.2 Garamon and GAALOP Compatibility}\label{garamon-and-gaalop-compatibility}
\addcontentsline{toc}{subsection}{8.2 Garamon and GAALOP Compatibility}

Generated geometric kernels provide a practical comparison for our native support analysis. Farscape\textquotesingle s foreign-binding path could import a generated kernel behind an explicit algebra, support, and numerical contract. A C or C++ ABI alone does not recover those facts. The integration would need either generated metadata or a separately checked wrapper, with tests and proofs tied to the generator version.

\subsection*{8.3 Joint Grade and Support Inference}\label{joint-grade-and-support-inference}
\addcontentsline{toc}{subsection}{8.3 Joint Grade and Support Inference}

The grade-axis research separates the declared algebra from two kinds of value information. Homogeneous parity composes in \(\mathbb{Z}_2\). Blade support is propagated over a finite powerset, with grade support as its derived image. A value such as \(1+e_1\) has both parities, so a total analysis needs a mixed-parity result or an equivalent set. A use requiring homogeneous parity creates an obligation to establish that property.

The purpose of propagating support is to carry the coefficient savings of Section 3 through sequences of operations. For a fixed algebra, addition combines the coefficient positions either input may contain. Multiplication examines which pairs of input components can contribute under the declared basis and metric. Operations such as the regressive product also require the algebra\textquotesingle s declared duality rules. Repeating these rules until the supports stop changing gives a stable account of which coefficients each result may need. That information can guide smaller representations and eliminate arithmetic for components proved to be zero, while preserving the selected algebra\textquotesingle s operations.

\emph{Decidable by Construction} (Sections 2.2 and 5.3) \hyperref[ref-haynes-dbc]{{[}5{]}} develops the account of separate annotation domains and checked exchanges between them. Here the task is to supply the geometric rules and their preservation arguments within that framework. A proved parity fact, for example, can narrow the possible blade support. A use-site requirement creates an obligation to prove that fact. Our next step is to define the supported combinations of constraints and establish principality, termination, and analysis-cost bounds for that fragment. This would allow richer geometric information to guide representation choices while preserving predictable inference and responsive design-time feedback.

Those rules also need operation-level contracts. Checked constructors and updates must establish the versor and normalization properties needed for a declared rotation. Covering the admitted primitives and representation conversions ensures that every permitted way of constructing or changing the value preserves its geometric meaning. The graded-modality literature supplies a basis for this work, including recent comparisons of graded-base and linear-base coeffect systems \hyperref[ref-liepelt-coeffect-2026]{{[}36{]}}.

This has direct bearing on the soundness of our approach to continuations and parallelism. A suspended computation must retain the algebra and support of captured values, together with the storage and permissions needed when it resumes. Parallel branches must preserve the same contracts when work is divided and results are combined. Our ongoing coexponential research develops the rules for expressing and composing these interactions for a planned companion paper. The PHG\textquotesingle s role is to retain their value, resource, and dependency premises so continuation and parallel-execution rules can be checked against the contracts established for the computation.

\subsection*{8.4 Conformal and Projective Algebra Selection}\label{conformal-and-projective-algebra-selection}
\addcontentsline{toc}{subsection}{8.4 Conformal and Projective Algebra Selection}

Choosing between projective geometric algebra (PGA) and conformal geometric algebra (CGA) affects both the geometric operations available to a program and the coefficient representation carried through computation. A valid alternative representation can require fewer coefficient slots and products, reducing storage and data movement. That can matter when fitting a computation onto a spatial tile or synthesizing its arithmetic on an FPGA. The savings depend on the operations and support the program actually uses.

For a computation whose objects and operations can be represented in both models, the first task is to establish a mapping that preserves their meaning, including normalization and numerical behavior. The declared target can then rank the admissible implementations by arithmetic, storage, and communication cost. Our grade-axis and inference research develops the distinct roles of algebra identity, parity, and blade support needed for these decisions, for a planned companion paper. An engineer could then keep the geometric contract stable while the tooling evaluates admissible implementations for a chosen target.

\subsection*{8.5 Quantum-Classical Compilation and the PHG}\label{quantum-classical-compilation-and-the-phg}
\addcontentsline{toc}{subsection}{8.5 Quantum-Classical Compilation and the PHG}

A stabilizer measurement is a prospective joint relation: several data qubits participate in one check with a declared Pauli product, preparation, gate sequence, and measurement interpretation. A circuit decomposition must preserve that operation. The PHG could retain the correspondence and its control and interchange requirements.

Unitarity, \(U^\dagger U=I\), is a matrix identity requiring its own evidence. A restricted gate library can compose proved unitary constructors; measurement additionally needs a probabilistic account and classical outcomes. The negative and fractional types paper \hyperref[ref-haynes-nft]{{[}8{]}} develops the algebraic setting. A future PHG integration must distinguish these ideal-operation contracts from noisy hardware behavior and from the numerical error of a classical simulation.

\subsection*{8.6 Neuromorphic Targets and Temporal Hyperedges}\label{neuromorphic-targets-and-temporal-hyperedges}
\addcontentsline{toc}{subsection}{8.6 Neuromorphic Targets and Temporal Hyperedges}

A temporal hyperedge can retain a coincidence condition: a specified collection of spikes arrives within a window \(\tau\), so \(\max_i t_i-\min_i t_i\leq\tau\). Its contract must identify the participating events, their multiplicity, and the selection policy for a larger stream. Pairwise inequalities can express the same condition for known timestamps. The grouped relation preserves the identity of the event being recognized.

The timestamp units and declared neuron dynamics can be checked statically; future arrivals remain runtime events. Existing tools already account for physical units, including Brian 2 \hyperref[ref-brian-units]{{[}37{]}} and the dimensional parameters documented by NEST \hyperref[ref-nest-iaf-psc-exp]{{[}38{]}}. Synaptic-weight units follow the model equation: a current amplitude and a conductance are different quantities.

Our proposed extension would retain these contracts through placement and heterogeneous transfer. An on-chip learning target needs explicit trace layouts, update semantics, and learning-engine capabilities. A fixed number of trace variables per synapse bounds the local state shape; total storage still depends on synapse count, precision, and lifetime. The companion \emph{Adaptive Domain Models} paper \hyperref[ref-haynes-adm]{{[}4{]}} treats adaptation and deployment. The PHG integration must establish that the chosen target operations preserve the declared model and update rules.

\subsection*{8.7 Grade-Structured Tangent Selection}\label{grade-structured-tangent-selection}
\addcontentsline{toc}{subsection}{8.7 Grade-Structured Tangent Selection}

A geometric parameterization may exclude coefficient slots entirely. Tangent directions need no component in those absent slots, so verified support can reduce the space explored by forward differentiation. A learned coefficient currently equal to zero is different: its derivative may be nonzero. Nor does sparse forward arithmetic imply fewer trainable directions; the diagonal computation \(y_i=w_ix_i\) can leave a scalar loss dependent on every \(w_i\).

Flügel et al. \hyperref[ref-flugel-multi-tangent-2026]{{[}39{]}} give a useful result for combining several directions. In Euclidean parameter coordinates, with declared reference scales for physically heterogeneous quantities and linearly independent tangent columns \(\mathbf V\) spanning \(U\),

\[P_U(\nabla f)=\mathbf V(\mathbf V^\top\mathbf V)^{-1}\mathbf V^\top\nabla f\]

is the closest vector to the gradient within \(U\), and recovers the gradient when it lies in that subspace. This optimization metric is distinct from the Clifford algebra\textquotesingle s metric, which may be indefinite or degenerate. Finite-precision evaluation also requires conditioning and linear-solve checks.

The PHG contribution would be to establish an admissible subspace from the actual parameter dependencies, then guide tangent selection within it. A support restriction does not establish nonlinear invariants such as motor normalization; the parameterization and update rule must preserve those too. \emph{Adaptive Domain Models}, Section 2.4 \hyperref[ref-haynes-adm]{{[}4{]}}, develops the estimator and residual limits, while \emph{Decidable by Construction}, Section 8.1 \hyperref[ref-haynes-dbc]{{[}5{]}}, treats reference scales and update units.

A probabilistic model adds a choice of optimization geometry. Our proposed extension binds that choice to the same parameter coordinates, admissible tangents, loss differential and model revision as the update. For tangent columns \(\mathbf V\) and a declared positive-semidefinite parameter metric \(G\), the restricted metric matrix is \(\mathbf V^\top G\mathbf V\). Its entries are obligations about the joint tangent and metric evaluation; the Euclidean Gram matrix \(\mathbf V^\top\mathbf V\) cannot replace it without establishing the required equality. The Clifford multiplication metric retains its separate algebraic role.

The companion ADM construction derives a restricted natural-gradient solve from this metric and directional loss measurements. Here its graph contract requires the trace law, score and output derivatives, observation policy and numerical evidence from Section 2.4 to agree at that solve. It also records the kernel condition needed for a degenerate metric, the admissibility of the resulting update, and the distinction between metric-estimation error, derivative error and solve error. These requirements extend our existing joint evidence; they do not assert a working statistical-geometry compiler pass. An experiment can test the proposed integration by changing one dependency and checking that exactly the affected evidence is reconsidered.

A useful experiment would compare tangent strategies at equal cost under a verified support restriction, measuring optimization quality and storage together. Structural admissibility supplies a reason to explore fewer directions; it does not by itself certify the quality of an individual update.

\subsection*{8.8 BAREWire Geometric Serialization}\label{barewire-geometric-serialization}
\addcontentsline{toc}{subsection}{8.8 BAREWire Geometric Serialization}

Our BAREWire substrate provides declarations for memory mapping, inter-process communication (IPC), and network interchange. The proposed geometric integration would connect a shared value contract to each endpoint\textquotesingle s local representation. A packed wire encoding and a naturally aligned application binary interface (ABI) layout can differ while implementing that contract.

For geometric interchange, the shared declaration would include algebra identity and basis order, together with the permitted blade support and numerical conversion rules. A value supported in the even subalgebra of \(\mathrm{Cl}(3,0,1)\) needs eight coefficient slots.

A rigid pose describes an object\textquotesingle s position and orientation relative to a coordinate frame. In PGA, a \emph{motor} combines the rotation and translation that establish that pose \hyperref[ref-ong-klein]{{[}16{]}}. A proposed checked constructor would build the motor from validated rotation and translation inputs under the declared coordinate conventions, establishing the required motor and normalization conditions. The resulting contract describes motion that preserves shape and distances, with numerical error governed by the declared arithmetic and conversion rules. A receiving endpoint could then apply the pose to geometry or compose it with another pose while preserving that agreed meaning.

The support requirement is containment: every coefficient the producer may supply must belong to the declared support. The tooling checks this automatically against the producer\textquotesingle s inferred support. If it cannot establish containment, it reports the potentially incompatible coefficient positions at design time and blocks compilation of that interchange. The diagnostic connects the producer\textquotesingle s possible output to the receiving declaration, so the developer can see which value or contract needs attention. Widening to the declared layout may insert known-zero slots. Values from incompatible algebra declarations require a checked mapping, and value-level refinements must survive any numerical conversion.

For the shared even-support layout, a Clef contract need not name a source-level width. The proposed geometric type vocabulary could express:

\begin{lstlisting}
module GeometricInterchange =
    [<Algebra(3, 0, 1)>]
    type PGA

    type EvenState = Multivector<PGA, even>

    type StateUpdate =
        { Body     : BodyId
          State    : EvenState
          Observed : float<s> }
\end{lstlisting}

Suppose the declared boundary resolves \lstinline!BodyId! to BARE \lstinline!u64! and both coefficient and time representations to BARE \lstinline!f64!. Its derived encoding could be:

\begin{lstlisting}
type StateUpdate struct {
  body: u64
  state: list<f64>[8]
  observed: f64
}
\end{lstlisting}

The BARE schema above is an output of boundary resolution: it records the selected field order, concrete widths, and byte encoding. The resulting representation is fully unboxed and untagged: the payload carries the field values directly, without object wrappers or runtime type tags. By itself, \lstinline!observed: f64! specifies how to encode a number. It does not establish that the number represents seconds or that the eight state coefficients satisfy a pose contract.

BAREWire is our implementation for carrying the framework\textquotesingle s dimensional contracts through memory and interchange mappings. It connects the chosen layout to source-level measures and target requirements, so the generated boundary code can be checked against the meaning of the values it transfers. For geometric interchange, the proposed extension would retain algebra identity, permitted support, and value invariants alongside that dimensional account. Both endpoints would then have a shared basis for checking the value\textquotesingle s meaning as well as its encoding.

The fixed field sequence above has a statically known packed extent. More general BARE encodings include variable-length fields \hyperref[ref-bare-standard]{{[}40{]}}. Schema agreement avoids repeating field names and type descriptions in each value. Decoding must enforce bounds, detect malformed inputs, and respect agreed versions. Zero-copy access requires compatible byte order, alignment, extent, and ownership.

BAREWire\textquotesingle s memory declarations are applied in the PSG with the target description before witness lowering. Extending that process to geometric IPC and network mappings requires evidence for the sender, codec, transport, and receiving conversion. External bytes need the validation assigned to the receiving boundary. Numeric fidelity must be checked alongside layout compatibility: transferring eight coefficients correctly does not prove that a change of format preserves their geometric invariant.

The consultation pattern in \emph{Adaptive Domain Models} \hyperref[ref-haynes-adm]{{[}4{]}} would use this interchange to send structured quantities to a specialist model and receive a structured result. Training must preserve the declared support and value invariants, and a posterior payload needs an agreed distribution representation and evidence semantics.

A geometric partition across CPU and FPGA is a useful first experiment: derive both mappings from the shared contract, prove their conversion conditions, and reconcile the PSG and lowered obligations. The same discipline can later extend to additional processors without assuming that every endpoint has an identical in-memory byte layout. It\textquotesingle s the contract that supports that heterogyny which allows the Fidelity Framework to support such a broad remit.

\section*{9. Shared Semantic Contracts}\label{shared-semantic-contracts}
\addcontentsline{toc}{section}{9. Shared Semantic Contracts}

A geometric computation can retain one contract through source editing, arithmetic selection, and transfer to another processor. That is the engineering result we seek from the PHG. Its explicit relations give each joint operation and resource requirement a place in the same proof-carrying graph.

The next implementation evidence should connect those stages for a small geometric program: establish its grade and support facts, select a numerically justified kernel, and preserve its BAREWire contract across a real boundary. The resulting reconciliation should identify every obligation by its source and show what each lowering witness discharged. That would make the architecture assessable where its promises have to hold: in the generated program.

\section*{Acknowledgments}\label{acknowledgments}
\addcontentsline{toc}{section}{Acknowledgments}

The author thanks John L. Gustafson for direct support on posit arithmetic and value distribution analysis. The treatment of quire-backed gradient accumulation in Section 5.2 and the posit representation selection for geometric algebra field quantities reflect his technical guidance. The author also thanks Paul Snively for his input on the geometric algebra problem space and for opening lines of investigation that this paper pursues. The author also thanks Martin Coll, whose work on interaction-net rewriting in MLIR informed the investigation of joint reduction in our own PSG described in Section 1.4.

\section*{Software Availability}\label{software-availability}
\addcontentsline{toc}{section}{Software Availability}

The Clef language, Composer compiler, and supporting libraries described in this paper are developed under the Fidelity Framework project. Source repositories are available at \url{https://github.com/FidelityFramework}. The language specification, compiler documentation, and ongoing project research are available at \url{https://clef-lang.com}. Licensing terms are detailed in each repository. They contain the compiler and BAREWire implementation alongside design material; the proposed geometric integrations in this paper identify work still to be demonstrated.

\section*{Appendix A. Grade Inference Example}\label{appendix-a.-grade-inference-example}
\addcontentsline{toc}{section}{Appendix A. Grade Inference Example}

Take three normalized finite points \(P_1,P_2,P_3\) in plane-based 3D PGA. Each point has grade 3. Their ordered join

\[F=P_1\vee P_2\vee P_3\]

has grade 1 when nonzero. In the convention of De Keninck et al. \hyperref[ref-de-keninck-mesh-2025]{{[}12{]}}, its magnitude gives twice the triangle area. The area is a scalar with physical dimension \(\mathrm{m}^2\), distinct from the carrier\textquotesingle s geometric grade. Coincident or collinear vertices produce a degenerate result whose handling belongs to the operation contract.

The binary form first computes the grade-2 line \(P_1\vee P_2\). The PHG form retains the three-point join as one operation. Both are algebraically valid. A fused implementation is permitted when its numeric behavior satisfies the same contract.

For a coordinate check, let \(u=x_2-x_1\) and \(v=x_3-x_1\) in \(\mathbb R^3\). Then

\[A=\tfrac12\sqrt{(u_yv_z-u_zv_y)^2+(u_zv_x-u_xv_z)^2+(u_xv_y-u_yv_x)^2}.\]

Each cross-product component has dimension \(\mathrm{m}^2\). The squared sum has dimension \(\mathrm{m}^4\), and the square root restores \(\mathrm{m}^2\). The expression requires coordinate differences, six cross-product multiplications, and three subtractions before the norm calculation. A three-term product is not automatically an exact quire input, and a square root is not an exact accumulation operation.

{\def\LTcaptype{none} 
\begin{longtable}[]{@{}
  >{\raggedright\arraybackslash}p{(\linewidth - 2\tabcolsep) * \real{0.5000}}
  >{\raggedright\arraybackslash}p{(\linewidth - 2\tabcolsep) * \real{0.5000}}@{}}
\toprule\noalign{}
\begin{minipage}[b]{\linewidth}\raggedright
Established fact
\end{minipage} & \begin{minipage}[b]{\linewidth}\raggedright
Consequence or remaining obligation
\end{minipage} \\
\midrule\noalign{}
\endhead
\bottomrule\noalign{}
\endlastfoot
Three normalized finite PGA points & Grade-1 join carrier, or zero in a degenerate case \\
Coordinates measured in metres & Scalar area measured in square metres \\
Declared bounds or a geometric invariant & A justified area range for representation selection \\
Requested numerical error bound & Check subtraction, accumulation, square root, and conversion \\
Target exact-accumulation capability & Prove partial-sum capacity where a quire is used \\
Selected representation and boundary mapping & Preserve the area contract through lowering and transfer \\
\end{longtable}
}

\emph{Table A.1: Separate geometric, dimensional, and numerical obligations for triangle area.}

A comparison of posit and IEEE arithmetic for this area calculation needs its input range, scaling convention, and error criterion. On a processor it compares available implementations. For instance with an FPGA, widths and arithmetic can be synthesized within the declared resource limits. This has consequence for the admission of optimized sub-processing and acceleration that can move hardware/software design forward at an accelerated pace.

\section*{Appendix B. Co-location Hyperedge Example}\label{appendix-b.-co-location-hyperedge-example}
\addcontentsline{toc}{section}{Appendix B. Co-location Hyperedge Example}

A four-stage pipeline loads data at \(A\), computes two partial results at \(B\) and \(C\), and reduces them at \(D\). Its placement relation retains the four stages and their resource contract:

\[f=((\mathrm{load}_A,\mathrm{compute}_B,\mathrm{compute}_C,\mathrm{reduce}_D),\mathrm{output},\lambda_{\mathrm{tile}}).\]

{\def\LTcaptype{none} 
\begin{longtable}[]{@{}
  >{\raggedright\arraybackslash}p{(\linewidth - 2\tabcolsep) * \real{0.5000}}
  >{\raggedright\arraybackslash}p{(\linewidth - 2\tabcolsep) * \real{0.5000}}@{}}
\toprule\noalign{}
\begin{minipage}[b]{\linewidth}\raggedright
Annotation field
\end{minipage} & \begin{minipage}[b]{\linewidth}\raggedright
Requirement
\end{minipage} \\
\midrule\noalign{}
\endhead
\bottomrule\noalign{}
\endlastfoot
Placement domain & Four declared sites within an allowed tile block \\
Routes & \(A\to B\), \(A\to C\), \(B\to D\), \(C\to D\) \\
Storage & Each live buffer fits its assigned memory, including aggregate use \\
DMA & Transfers use available channels and descriptor capacity \\
Synchronization & \(D\) consumes each input after its producing stage completes \\
\end{longtable}
}

\emph{Table B.1: Joint placement and communication requirements.}

A proposed \(2\times2\) assignment places \(A\) at \((0,0)\), \(B\) at \((0,1)\), \(C\) at \((1,0)\), and \(D\) at \((1,1)\). The target\textquotesingle s declared interconnect determines whether these logical routes and concurrent transfers are feasible. The assignment alone supplies no direct DMA channel between arbitrary tiles.

Pairwise equalities on a common placement variable exactly express one shared tile. Pairwise adjacency is a different predicate, and neither predicate by itself checks aggregate buffer capacity. The hyperedge records which requirement is intended and preserves it when the pipeline is partitioned. Lowering can consume the assignment after its capacity, routing, and synchronization obligations have been established.

\section*{References}\label{references}
\addcontentsline{toc}{section}{References}

\leavevmode\protect\phantomsection\label{ref-haynes-dts-dmm}{}{[}1{]} Haynes, H. (2026). Dimensional type systems and deterministic memory management: Design-time semantic preservation in native compilation. \href{https://arxiv.org/abs/2603.16437}{arXiv:2603.16437}.

\leavevmode\protect\phantomsection\label{ref-rico-xdna-2024}{}{[}2{]} Rico, A., et al. (2024). AMD XDNA NPU in Ryzen AI processors. \emph{IEEE Micro}, 44(6):73-\/-83.

\leavevmode\protect\phantomsection\label{ref-kennedy-units-2009}{}{[}3{]} Kennedy, A. (2009). Types for units-of-measure: Theory and practice. In \emph{Central European Functional Programming School}, LNCS 6299. Springer.

\leavevmode\protect\phantomsection\label{ref-haynes-adm}{}{[}4{]} Haynes, H. (2026). Adaptive domain models: Bayesian evolution, warm rotation, and principled training for geometric and neuromorphic AI. \href{https://arxiv.org/abs/2603.18104}{arXiv:2603.18104}.

\leavevmode\protect\phantomsection\label{ref-haynes-dbc}{}{[}5{]} Haynes, H. (2026). Decidable by construction: Design-time verification for trustworthy AI. \href{https://arxiv.org/abs/2603.25414}{arXiv:2603.25414}.

\leavevmode\protect\phantomsection\label{ref-lattner-mlir-2021}{}{[}6{]} Lattner, C., et al. (2021). MLIR: Scaling compiler infrastructure for domain specific computation. In \emph{Proceedings of CGO}.

\leavevmode\protect\phantomsection\label{ref-karwowski-program-geometry}{}{[}7{]} Karwowski, J., Nielsen, F., and Staton, S. (2026). Compositional information geometry of probabilistic programs. \href{https://arxiv.org/abs/2610.17886v1}{arXiv:2610.17886v1}.

\leavevmode\protect\phantomsection\label{ref-haynes-nft}{}{[}8{]} Haynes, H. (2026). Negative and fractional types in the Fidelity Framework. \href{https://arxiv.org/abs/2606.04352}{arXiv:2606.04352}.

\leavevmode\protect\phantomsection\label{ref-dorst-pga-2022}{}{[}9{]} Dorst, L., and De Keninck, S. (2022). A guided tour to the plane-based geometric algebra PGA, version 2.0. \href{https://bivector.net/PGA4CS.html}{Tutorial}.

\leavevmode\protect\phantomsection\label{ref-lounesto-clifford-2001}{}{[}10{]} Lounesto, P. (2001). \emph{Clifford Algebras and Spinors}, 2nd edition. London Mathematical Society Lecture Note Series 286. Cambridge University Press.

\leavevmode\protect\phantomsection\label{ref-zhdanov-flash-clifford}{}{[}11{]} Zhdanov, M. (2025). Flash Clifford: Hardware-efficient implementation of Clifford algebra neural networks. \href{https://github.com/maxxxzdn/flash-clifford\#performance}{Project and performance notes}.

\leavevmode\protect\phantomsection\label{ref-de-keninck-mesh-2025}{}{[}12{]} De Keninck, S., Roelfs, M., Dorst, L., and Eelbode, D. (2025). Clean up your mesh! Part 1: Plane and simplex. \emph{arXiv preprint arXiv:2511.08058}.

\leavevmode\protect\phantomsection\label{ref-shewchuk-predicates-1997}{}{[}13{]} Shewchuk, J. R. (1997). Adaptive precision floating-point arithmetic and fast robust geometric predicates. \emph{Discrete \& Computational Geometry}, 18:305-363. \href{https://www.cs.cmu.edu/~quake/robust.html}{Paper and software}.

\leavevmode\protect\phantomsection\label{ref-bivector-catalog}{}{[}14{]} bivector.net. (2025). biVector.net geometric algebra library catalog. \href{https://bivector.net/lib.html}{Library catalog}.

\leavevmode\protect\phantomsection\label{ref-versor}{}{[}15{]} Colapinto, P. (2014). Versor: A fast, lightweight C++ library for geometric algebra. University of California Santa Barbara. \href{https://wolftype.github.io/versor/}{Project}.

\leavevmode\protect\phantomsection\label{ref-ong-klein}{}{[}16{]} Ong, J. (2020). Klein: A fast geometric algebra library for 3D PGA. \href{https://github.com/jeremyong/Klein}{Project}.

\leavevmode\protect\phantomsection\label{ref-amd-xdna-architecture}{}{[}17{]} AMD. \emph{AMD XDNA Architecture}. \href{https://www.amd.com/en/technologies/xdna.html}{Architecture overview}.

\leavevmode\protect\phantomsection\label{ref-karypis-partitioning-2000}{}{[}18{]} Karypis, G., and Kumar, V. (2000). Multilevel \(k\)-way hypergraph partitioning. \emph{VLSI Design}, 11(3):285-\/-300.

\leavevmode\protect\phantomsection\label{ref-mlir-aie}{}{[}19{]} AMD/Xilinx. (2024). MLIR-AIE: An MLIR-based toolchain for AMD AI engines. \href{https://github.com/Xilinx/mlir-aie}{Project}.

\leavevmode\protect\phantomsection\label{ref-cuda-programming-model}{}{[}20{]} NVIDIA. Programming model. \emph{CUDA Programming Guide}. \href{https://docs.nvidia.com/cuda/cuda-programming-guide/01-introduction/programming-model.html\#warps-and-simt}{Warps and SIMT}.

\leavevmode\protect\phantomsection\label{ref-mlir-nvvm-dialect}{}{[}21{]} MLIR. \emph{\textquotesingle nvvm\textquotesingle{} Dialect}. \href{https://mlir.llvm.org/docs/Dialects/NVVMDialect/}{Dialect documentation}.

\leavevmode\protect\phantomsection\label{ref-raissi-pinn-2019}{}{[}22{]} Raissi, M., Perdikaris, P., and Karniadakis, G. E. (2019). Physics-informed neural networks: A deep learning framework for solving forward and inverse problems involving nonlinear PDEs. \emph{Journal of Computational Physics}, 378:686-\/-707.

\leavevmode\protect\phantomsection\label{ref-posit-standard-2022}{}{[}23{]} Posit Working Group. (2022). \emph{Posit Standard}. \href{https://posithub.org/docs/posit_standard-2.pdf}{Standard}.

\leavevmode\protect\phantomsection\label{ref-baydin-forward-gradients-2022}{}{[}24{]} Baydin, A. G., Pearlmutter, B. A., Syme, D., Wood, F., and Torr, P. (2022). Gradients without backpropagation. \emph{arXiv preprint arXiv:2202.08587}.

\leavevmode\protect\phantomsection\label{ref-shazeer-moe-2017}{}{[}25{]} Shazeer, N., Mirhoseini, A., Maziarz, K., Davis, A., Le, Q., Hinton, G., and Dean, J. (2017). Outrageously large neural networks: The sparsely-gated mixture-of-experts layer. \emph{arXiv preprint arXiv:1701.06538}.

\leavevmode\protect\phantomsection\label{ref-harper-phase-distinction-1990}{}{[}26{]} Harper, R., Mitchell, J. C., and Moggi, E. (1990). Higher-order modules and the phase distinction. In \emph{Proceedings of the 17th ACM SIGPLAN-SIGACT Symposium on Principles of Programming Languages} (POPL), pages 341--354. ACM.

\leavevmode\protect\phantomsection\label{ref-robertson-seymour-treewidth-1986}{}{[}27{]} Robertson, N., and Seymour, P. D. (1986). Graph minors. II. Algorithmic aspects of tree-width. \emph{Journal of Algorithms}, 7(3):309--322. \href{https://doi.org/10.1016/0196-6774(86)90023-4}{doi:10.1016/0196-6774(86)90023-4}.

\leavevmode\protect\phantomsection\label{ref-gavril-chordal-1974}{}{[}28{]} Gavril, F. (1974). The intersection graphs of subtrees in trees are exactly the chordal graphs. \emph{Journal of Combinatorial Theory, Series B}, 16(1):47-56. \href{https://doi.org/10.1016/0095-8956(74)90094-X}{doi:10.1016/0095-8956(74)90094-X}.

\leavevmode\protect\phantomsection\label{ref-ruhe-equivariant-2023}{}{[}29{]} Ruhe, D., Brandstetter, J., and Forré, P. (2023). Clifford group equivariant neural networks. \emph{arXiv preprint arXiv:2305.11141}.

\leavevmode\protect\phantomsection\label{ref-zhdanov-steerable-2024}{}{[}30{]} Zhdanov, M., Ruhe, D., Weiler, M., Lucic, A., Brandstetter, J., and Forré, P. (2024). Clifford-steerable convolutional neural networks. In \emph{Proceedings of ICML}.

\leavevmode\protect\phantomsection\label{ref-qian-sessions-2021}{}{[}31{]} Qian, Z., Kavvos, G. A., and Birkedal, L. (2021). Client-server sessions in linear logic. \emph{arXiv:2010.13926}.

\leavevmode\protect\phantomsection\label{ref-verse-calculus-2023}{}{[}32{]} Peyton Jones, S., Claessen, K., Jhala, R., Sweeney, T., Augustsson, L., and Shivers, O. (2023). The Verse calculus: A core calculus for functional logic programming. \emph{Proceedings of the ACM on Programming Languages}, 7(ICFP).

\leavevmode\protect\phantomsection\label{ref-petricek-coeffects-2014}{}{[}33{]} Petricek, T., Orchard, D., and Mycroft, A. (2014). Coeffects: A calculus of context-dependent computation. In \emph{Proceedings of ICFP}.

\leavevmode\protect\phantomsection\label{ref-hmetis-manual}{}{[}34{]} Karypis, G., and Kumar, V. (1998). \emph{hMETIS: A Hypergraph Partitioning Package}, version 1.5.3. University of Minnesota. \href{https://karypis.github.io/glaros/files/sw/hmetis/manual.pdf}{User manual}.

\leavevmode\protect\phantomsection\label{ref-patoh-manual}{}{[}35{]} Çatalyürek, Ü. V., and Aykanat, C. (2011). \emph{PaToH: Partitioning Tool for Hypergraphs}. User manual, revised March 2011. \href{https://faculty.cc.gatech.edu/~umit/PaToH/manual.pdf}{Manual}.

\leavevmode\protect\phantomsection\label{ref-liepelt-coeffect-2026}{}{[}36{]} Liepelt, V., Marshall, D., and Orchard, D. (2026). Same coeffect, different base: Connecting two dominant approaches to graded types. \href{https://arxiv.org/abs/2606.28042}{arXiv:2606.28042}.

\leavevmode\protect\phantomsection\label{ref-brian-units}{}{[}37{]} Brian team. Physical units. Brian 2 documentation. \href{https://brian2.readthedocs.io/en/stable/user/units.html}{Documentation}.

\leavevmode\protect\phantomsection\label{ref-nest-iaf-psc-exp}{}{[}38{]} NEST Simulator Documentation. \emph{iaf\_psc\_exp: Leaky integrate-and-fire neuron model with exponential-shaped input currents}. \href{https://nest-simulator.readthedocs.io/en/stable/models/iaf_psc_exp.html}{Model reference}.

\leavevmode\protect\phantomsection\label{ref-flugel-multi-tangent-2026}{}{[}39{]} Flügel, K., Coquelin, D., Weiel, M., Debus, C., Streit, A., and Götz, M. (2026). Beyond backpropagation: Optimization with multi-tangent forward gradients. \emph{arXiv preprint arXiv:2410.17764}. Revised January 2026.

\leavevmode\protect\phantomsection\label{ref-bare-standard}{}{[}40{]} BARE contributors. (2026). Binary Application Record Encoding. Internet-Draft draft-devault-bare-15. \href{https://datatracker.ietf.org/doc/draft-devault-bare/}{Internet-Draft}.
\end{document}